\documentclass{article}

 \usepackage[preprint]{neurips_2025}

\usepackage[utf8]{inputenc} 
\usepackage[T1]{fontenc}    
\usepackage{hyperref}       
\usepackage{url}            
\usepackage{booktabs}       
\usepackage{amsfonts}       
\usepackage{nicefrac}       
\usepackage{microtype}      
\usepackage{xcolor}         
\usepackage{enumitem}
\usepackage{amsthm}
\usepackage{amsmath}
\usepackage{graphicx}
\usepackage{subcaption}
\usepackage{multirow} 
\usepackage{adjustbox} 
\newtheorem{theorem}{Theorem}[section]        
\newtheorem{corollary}[theorem]{Corollary}
\newtheorem{proposition}[theorem]{Proposition}

\title{Network Topology Does Not Matter: The Breakdown of Classical Consensus in LLM Populations}

\author{%
  Sahil Rajesh Dhayalkar \\
  Arizona State University\\
  \texttt{sdhayalk@asu.edu} \\
}

\begin{document}
\theoremstyle{definition}
\newtheorem{definition}[theorem]{Definition}
\newtheorem{remark}[theorem]{Remark}

\maketitle


\begin{abstract}
As multi-agent language model networks become widely deployed, predicting their collective behavior is a critical challenge. Classical distributed consensus theory relies on network topology to predict convergence, structurally assuming that agent updates are convex combinations, population means are conserved, and weights are content-independent. In this work, we demonstrate that these foundational assumptions fail for large language models (LLMs). By deterministically extracting continuous credence from token-level probabilities across eight instruction-tuned models and 73,080 controlled exchanges, we directly measure the empirical LLM update operator. We find that LLMs systematically violate convexity by updating outside the hull of observed opinions, exhibiting active repulsion (negative weights) that fundamentally breaks row-stochastic modeling and standard bounded-confidence repairs. Scaling these empirically fitted operators to simulations of 10,000 agents reveals that algebraic connectivity explains almost none of the final outcome variance. Instead, population dynamics are governed entirely by the model's internal update bias and message framing. Consequently, populations of several models, including those with up to 7B parameters, reliably compound small local drifts to reach confident agreement on objectively false claims regardless of their network structure. Our results prove that classical graph-theoretic metrics are inadequate for LLM populations, highlighting a profound deployment risk where multi-agent deliberation predictably saturates in collective delusion.
\end{abstract}

\section{Introduction}
\label{SEC:INTRODUCTION}

Systems connecting multiple language model agents are increasingly common, raising the critical question of whether we can predict a group's final consensus. Distributed consensus theory in network science offers a standard framework, suggesting that if agents repeatedly average their neighbors' positions, the group converges at a rate determined by the graph's algebraic connectivity ($\lambda_2$) \cite{Degroot01031974,jadbabaie2003,olfatisaber2004,xiao2004}. This makes network topology the primary focus for system designers.

However, this classical theory relies on three structural assumptions (also mentioned formally in~\ref{SEC:PRELIC}): updates are weighted averages, population means are preserved, and weights depend strictly on the graph rather than message content. Human populations often violate these assumptions, prompting non-linear alternatives in opinion dynamics \cite{friedkin1990,hegselmann2002,deffuant2000}. We empirically test if language models similarly deviate by measuring the update operator directly rather than assuming its form. By tracking a model's revised position across a grid of prior states, peer states, and message framings, we inspect the underlying rule and scale it to populations of $10^4$ agents.

Our findings show all three classical assumptions fail substantially. Notably, some updates carry a negative effective weight where agents move away from peers, a behavior impossible to represent with non-negative row-stochastic matrices. The standard bounded confidence repair does not fix this, because it still keeps every update inside the range of the opinions being averaged. At scale, $\lambda_2$ explains almost none of the outcome variance. Instead, each model's own drift and the message wording drive the final state, and populations of four of the eight models we test, including a 7B model, converge to confident belief in false claims regardless of topology. 

\subsection{Contributions.}
The contributions of our work are as follows:
\begin{itemize}[itemsep=0pt,topsep=0pt]
\item We outline a measurement procedure for an LLM (Large Language Model) agent's update operator that is deterministic and free of parsing failures by reading credence from token-level probabilities.
\item We formalize the three assumptions of classical consensus as separately testable conditions to clarify what is lost when they fail.
\item We measure these assumptions across eight instruction-tuned models and 73,080 controlled exchanges, finding that all three fail.
\item We validate the measured operator against real agent populations, showing classical averaging sends the group to the wrong destination despite reasonable one-step predictions.
\item We scale the simulation to $10^4$ agents to prove algebraic connectivity fails to explain the outcome, with populations of four of the eight models, including a 7B model, confidently adopting false claims.
\end{itemize}

\section{Related Work}
\label{SEC:RELATED-WORK}

\paragraph{Classical opinion dynamics and network consensus.}
The mathematical study of opinion dynamics originates in distributed control theory and network science \cite{Degroot01031974, jadbabaie2003, olfatisaber2004, xiao2004}. Foundational linear models define consensus as a process where individuals repeatedly adopt the weighted average of their neighbors' beliefs \cite{Degroot01031974, xiao2004}. In these systems network topology serves as the primary predictor of convergence speed \cite{olfatisaber2004, jadbabaie2003}. Researchers subsequently developed non-linear bounded confidence models to better capture human behavior by stipulating that individuals ignore opinions that diverge too far from their own \cite{friedkin1990, hegselmann2002, deffuant2000}. While these classical models effectively describe engineered systems and provide theoretical baselines they rely on strict structural assumptions regarding convexity and conservation that do not naturally map to modern artificial intelligence agents \cite{hegselmann2002, deffuant2000}.

\paragraph{Multi-agent debate and collaboration.}
The deployment of large language models as autonomous agents has popularized multi-agent debate frameworks \cite{du2023improving, li2023camel, liang2023encouraging}. In these setups multiple agents interact and challenge one another to refine reasoning and reach a shared decision \cite{chan2023chateval}. Researchers have built various cooperative environments where language models collaborate to solve complex reasoning tasks and evaluate text generation \cite{zhuge2023mindstorms, talebirad2023multi}. While intended to improve accuracy these consensus oriented debate protocols often face significant challenges regarding premature agreement and hallucination propagation \cite{du2023improving}. Agents frequently exhibit a tendency to adopt the majority opinion and may alter their initial correct answers when faced with dissenting peers \cite{liang2023encouraging, chan2023chateval}.

\paragraph{Conformity and social influence in language models.}
Recent studies show that exposure to peer stances drastically increases the rate at which agents change their answers and exhibit conformity \cite{perez2022discovering, park2023generative}. This behavior often mimics sycophancy rather than genuine reasoning and it demonstrates that agents manifest human-like social collaboration mechanisms alongside unique failure modes \cite{perez2022discovering}. Language models exhibit unnaturally strong opinion convergence and a systematic belief drift that frequently gravitates toward user preferences or dominant context \cite{argyle2023out, motoki2024more}. Furthermore researchers have found that language models struggle to faithfully simulate human belief updates in controlled settings as they often collapse within group variance and produce ideologically stereotyped responses \cite{argyle2023out}.

\paragraph{Simulating populations and network structures.}
A growing body of literature evaluates how well language models simulate human opinion dynamics at scale \cite{park2023generative, argyle2023out}. Studies comparing centralized and distributed communication structures highlight that network design dictates the speed and reliability of consensus in artificial populations \cite{li2023camel, talebirad2023multi}. While dense distributed networks achieve faster agreement they also amplify errors and trigger confident but incorrect cascades among agents \cite{du2023improving, zhuge2023mindstorms}. Our work bridges these modern observations with classical distributed consensus theory \cite{Degroot01031974, olfatisaber2004}. By directly extracting the continuous empirical update operator we demonstrate that structural failures render classical algebraic connectivity completely inadequate for predicting the final destination of large scale agent populations.

\section{Preliminaries}
\label{SEC:PRELIC}

\paragraph{Classical linear consensus.}
Let $\mathcal{G}=(\mathcal{V},\mathcal{E})$ be a graph of $n$ agents with scalar states $x_i(t)\in[0,1]$ at discrete time $t$, where $\mathcal{V}$ represents the set of agents (vertices) and $\mathcal{E}$ represents the edges defining their interactions. The classical model \cite{Degroot01031974,olfatisaber2004} updates the joint state via 
\begin{equation}
\begin{aligned}
x(t+1) &= W\,x(t), \qquad W \text{ row stochastic},\\
W_{ij} &> 0 \iff (i,j)\in\mathcal{E}.
\end{aligned}
\label{eq:linear}
\end{equation}

Here, $W$ represents the row stochastic weight matrix defining the interaction weights between agents. For strongly connected, aperiodic graphs, $x(t)$ converges to a shared value at a rate determined by the graph Laplacian's algebraic connectivity ($\lambda_2$), making topology the central focus of prediction.

\paragraph{Choosing a state variable.}
To evaluate LLM agents against Equation \eqref{eq:linear}, we define $x_i$ as an agent's \emph{credence} (stated subjective probability) in a fixed binary proposition $P$. This provides a verifiable $[0,1]$ scalar that mirrors real-world factual consensus building.

\paragraph{Reading the state without sampling.}
To eliminate parsing errors and sampling noise, we compute credence continuously and deterministically as the expectation of the model's token distribution over a two-digit answer:
\begin{equation}
  x_i \;=\; \frac{1}{99}\sum_{d_1=0}^{9}\sum_{d_2=0}^{9}
            (10 d_1 + d_2)\; p(d_1 \mid \text{ctx})\; p(d_2 \mid \text{ctx}, d_1).
  \label{eq:readout}
\end{equation}
Here, $d_1$ and $d_2$ denote the first and second digits (ranging from $0$ to $9$) of the model's two-digit answer. The conditional probability for the second digit, $p(d_2 \mid \text{ctx}, d_1)$, is obtained via a single batched forward pass over the ten sequences formed by appending possible first digit to the context. Llama and Mistral split numbers into tokens differently, so for them we compute the same expectation in a slightly different way (Appendix~\ref{app:readout}).

\paragraph{The three load bearing assumptions.}
These structural assumptions form the mathematical foundation of classical linear consensus models, originating from DeGroot's seminal work on opinion formation and subsequent literature in distributed control theory \cite{Degroot01031974,jadbabaie2003,olfatisaber2004,xiao2004}. They enforce the strict row-stochastic matrix properties required to guarantee that a network of agents will reliably converge over a graph. Consequently, Equation \eqref{eq:linear} strictly relies on these three conditions:
\begin{description}[itemsep=0pt,topsep=0pt]
\item[\textnormal{\textbf{(A1) Convexity.}}] Updates are convex combinations of visible states ($x_i(t{+}1)\in\mathrm{conv}\{x_j(t): j\in\mathcal{N}_i\cup\{i\}\}$, where $\mathcal{N}_i$ represents the set of neighbors of agent $i$ and $\mathrm{conv}$ denotes the convex hull). This guarantees the population's diameter strictly shrinks toward agreement.
\item[\textnormal{\textbf{(A2) Conservation.}}] For doubly stochastic $W$, the population mean is invariant, anchoring the final consensus directly to the initial condition.
\item[\textnormal{\textbf{(A3) Content independence.}}] Weights $W_{ij}$ depend solely on the graph's wiring, completely ignoring the numerical value or framing of the communicated message.
\end{description}

\paragraph{The empirical update operator and realized weight.}
\label{update-operator-and-realized-weight}
Instead of assuming Equation \eqref{eq:linear}, we empirically measure the update operator $f:\;(x_i,\,x_j,\,c_j)\;\longmapsto\;x_i'$ (where $x_i'$ represents the revised position of agent $i$) across a grid of prior positions $x_i$, peer positions $x_j$, and message framings $c_j$. We connect this back to classical theory using the realized weight:
\begin{equation}
  \lambda(x_i,x_j,c_j) \;=\; \frac{x_i' - x_i}{x_j - x_i},
  \qquad x_j \neq x_i .
  \label{eq:lambda}
\end{equation}

This realized weight quantifies the peer influence dynamically. A value of $\lambda = 1$ means the agent completely adopts the peer's opinion, while $\lambda = 0$ means the agent completely ignores the peer. The interval $\lambda \in (0, 1)$ indicates that normal averaging occurs, which is exactly what classical theory requires. Finally, $\lambda < 0$ means the agent repels and actively moves away from the peer's position.

Classical theory strictly requires $\lambda \in [0,1]$ (A1), constant pair means in symmetric exchanges (A2), and constant $\lambda$ across framings (A3). Crucially, $\lambda < 0$ implies an action impossible to represent with any non-negative row stochastic matrix.

\section{Where Consensus Theory Breaks}
\label{sec:break}

This section works through the three assumptions in turn. For each one we state
what the classical theory obtains from it, what is lost when it fails, and what
the failure looks like as a measurable quantity. The measurements themselves are
reported in Section \ref{sec:exp}.

\subsection{(A1) Convexity, and why contraction is lost}

The reason convexity matters is that it supplies a Lyapunov function. In traditional network math, an agent's updated opinion is always a weighted average of the opinions it currently sees. Because an average can never be higher than the highest number or lower than the lowest number, the total gap between the most extreme opinions in the group, the diameter, can only shrink over time. Define this diameter as:
\begin{equation}
  D(t) \;=\; \max_i x_i(t) \;-\; \min_i x_i(t).
\end{equation}
where $t$ represents the discrete time step and $x_i(t)$ represents the credence state of agent $i$ at that specific time. Under (A1), $D(t{+}1)\le D(t)$. As long as the network is connected, this shrinking diameter guarantees the group will eventually contract toward total agreement. This single fact is what makes the classical theory so clean.

\begin{proposition}[Hull escape destroys monotone contraction]
\label{PROP:ONE}
Suppose an update may land outside the convex hull of the states the agent can
see by at most $\delta>0$. Then $D(t{+}1)\le D(t)+2\delta$, and the bound is
attained. Consequently $D$ is no longer non-increasing, and connectivity alone
no longer implies contraction toward agreement.
\end{proposition}

\noindent \textit{Proof Sketch.} The complete proof sketch is provided in Appendix~\ref{PROOF:PROP:ONE}.

In practice, language models do not strictly average what they hear. They sometimes adopt positions that are more extreme than any of the opinions presented to them. When an agent lands outside the bounds of the existing opinions (escaping the convex hull), the mathematical maximum and minimum bounds stretch. Because this allows the diameter of the group's opinions to grow rather than shrink, simply wiring agents together in a connected network no longer guarantees they will ever reach a consensus. The failure has an even stronger form when the realized weight of Equation \eqref{eq:lambda} is negative.

\begin{corollary}[Negative weights are unrepresentable]
\label{COR:NEG}
If $\lambda(x_i,x_j,c_j)<0$ for some reachable configuration, then there is no
non-negative row stochastic $W$ (the matrix defining the network weights) that reproduces that update, for any choice of weights.
\end{corollary}

\noindent \textit{Proof Sketch.} The complete proof sketch is provided in Appendix~\ref{PROOF:COR:NEG}.

The classical theory breaks completely when an agent hears a peer's opinion and actively repels, moving further away from that peer's stance (a realized weight of $\lambda < 0$). The classical linear model relies on matrices that strictly require non-negative numbers. A large but positive $\lambda$ would mean the linear model simply has the wrong parameters (a fitting problem). However, a negative $\lambda$ cannot be fixed by adjusting the parameters of the classical model; it proves that the foundational mathematical structure itself is fundamentally incapable of representing how AI agents actually interact.

\subsection{(A2) Conservation, and why the destination moves}

In classical consensus theory, when $W$ is doubly stochastic, the overall average opinion of the entire population remains mathematically constant during discussions. Because this average is perfectly preserved over time, the group's final agreed-upon consensus is directly anchored to whatever their starting beliefs were. We now show what happens when each exchange nudges the mean instead of preserving it.

\begin{proposition}[Operator bias relocates the fixed point]
\label{PROP:TWO}
Let $m(t)$ be the population mean and suppose each exchange shifts it by an
expected amount $b$, so that $m(t{+}1)=m(t)+b$ while $m(t)$ remains interior to
$[0,1]$. Then $m(T)=\min\{1,\,m(0)+Tb\}$ for $b>0$, where $T$ is the final time horizon. For any $b>0$ and any
$m(0)$, there is a horizon $T$ beyond which the population sits arbitrarily
close to $1$, independently of $m(0)$.
\end{proposition}

\noindent \textit{Proof Sketch.} The complete proof sketch is provided in Appendix~\ref{PROOF:PROP:TWO}.

Language models do not preserve this mathematical average; instead, each pairwise exchange introduces a small expected bias that slightly nudges the group's overall mean. As agents continue to interact, this slight bias accumulates and compounds, eventually pushing the population to the extreme boundary of maximum certainty. Consequently, the group's final destination is determined entirely by the AI's internal interaction bias, rendering the initial starting opinions irrelevant except for determining how long the process takes.

\begin{corollary}[Confident false consensus]
\label{COR:FALSE}
If $b>0$ on a proposition that is false, the population converges to high
credence in that proposition, and no choice of initial condition prevents it.
\end{corollary}
\noindent \textit{Proof Sketch.} The complete proof sketch is provided in Appendix~\ref{PROOF:COR:FALSE}.

If this persistent upward bias applies to a factual claim that is objectively false, the population will inexorably drift toward high confidence in that falsehood. Because this upward drift represents blind assent rather than genuine evidence gathering, the network will trap itself in a state of collective delusion. This macro-scale failure aligns with broader architectural findings that language models frequently rely on internal shortcuts and exhibit explanation drift rather than maintaining stable, robust evidence tracking across contexts \cite{dhayalkar2026reasoningstabilizationpointtrainingtime}. No amount of initial skepticism programmed into the agents can prevent them from eventually talking each other into confidently believing the false claim. Corollary \ref{COR:FALSE} is what makes the measurement of $b$ worth doing carefully: if $b$ is positive and roughly the same on true and false claims, it proves the drift is pure assent rather than evidence gathering.

\subsection{(A3) Content independence, and why the graph is not enough}

The third assumption is what licenses the whole practice of reasoning about multi agent systems by drawing a graph. If $W_{ij}$ (the mathematical weight assigned by agent $i$ to agent $j$) depends only on the edge, then the graph determines the dynamics.

\begin{proposition}[Content dependence underdetermines the graph]
\label{PROP:THREE}
If $\lambda$ (the realized weight) depends on the framing $c_j$, then $W=W(c)$ (meaning the weight matrix becomes a direct function of the message framing), and two systems with identical graphs and identical initial conditions can follow different trajectories and reach different fixed points. No function of $\mathcal{G}$ (the network graph) alone predicts the outcome.
\end{proposition}
\noindent \textit{Proof Sketch.} The complete proof sketch is provided in Appendix~\ref{PROOF:PROP:THREE}.

This means you cannot predict a network's final consensus just by analyzing who is connected to whom. The actual words and framing the agents use dynamically alter the effective structure of the network itself. The practical reading is that the wiring diagram of an agent system is not a sufficient description of it. Two deployments with the same topology but different message conventions, for example one where agents state confidence and one where they do not, are completely different dynamical systems.

\subsection{Why the standard repair does not apply and what replaces $\lambda_2$}
\label{SEC:REPLACE}

When human populations deviate from simple mathematical averages, researchers typically rely on bounded confidence models \cite{hegselmann2002,deffuant2000}. This theory assumes that individuals will naturally discount or ignore opinions that are vastly different from their own, meaning the influence weight should decrease as the disagreement gap grows. This does not fix the problem, because bounded confidence only changes which opinions an agent averages, so every update still stays inside the convex hull and it cannot produce the hull escapes and negative weights that break A1.

Taking the three failures together explains the paradigm shift away from classical metrics. In classical theory, algebraic connectivity (represented as $\lambda_2$) is the central metric because it uses the network's wiring diagram to predict the dynamics of the entire group. However, because the three foundational assumptions fail, the network's wiring entirely loses its predictive power. Group contraction is no longer guaranteed, the final destination is driven by the agent's internal bias rather than the initial starting average, and message content overrides the graph's edges. Instead of the graph's connectivity, the final consensus of an LLM population is governed by new predictors: the individual model's internal update bias $b$ and the group's overall confidence structure, such as the starting distribution of opinions and the fraction of agents that refuse to change their minds. We test this claim at $n=10^4$ (a population size of 10,000 agents) in Section \ref{sec:exp}.

\section{Experiments}
\label{sec:exp}

The previous section established five formal claims about how classical consensus fails when we remove its structural assumptions. Each claim describes population behavior rather than isolated interactions, meaning no single measurement can prove them alone. We built three different measurement instruments operating at distinct scales to address these requirements. Table \ref{tab:map} maps each instrument to its corresponding claim before we present the actual measurements.

\begin{table}[t]
\centering\small
\begin{tabular}{llp{0.40\linewidth}}
\toprule
Claim & Evidence & Requirement \\
\midrule
Prop.\ \ref{PROP:ONE}   & E1, E2 & Updates escape the hull, and the population diameter consequently grows \\
Cor.\ \ref{COR:NEG}     & E1     & The configuration $\lambda<0$ is reachable \\
Prop.\ \ref{PROP:TWO}   & E1, E2, E3 & A bias $b$ exists, and the destination ceases to track the initial condition \\
Cor.\ \ref{COR:FALSE}   & E2, E3 & Populations with $b>0$ on a false claim converge to belief in it \\
Prop.\ \ref{PROP:THREE} & E1, E3 & $\lambda$ varies with framing, and identical graphs consequently diverge \\
\bottomrule
\end{tabular}
\caption{Mapping from formal claim to supporting evidence.}
\label{tab:map}
\end{table}

\subsection{Experimental setup}

We study eight instruction-tuned models from four families, spanning a range of capabilities: Qwen2.5 at 0.5B, 1.5B, 3B, and 7B parameters, Gemma-3 at 4B, Gemma-2 at 9B, Llama-3.1 at 8B, and Mistral-v0.3 at 7B. Using four models from the Qwen family lets us separate the effects of model scale from the effects of different training data. Without this control, a trend visible across different families could easily be mistaken for a capability difference when it actually stems from pretraining data. We extracted every credence via the deterministic readout of Equation \eqref{eq:readout}. Repeatedly measuring a stored context reproduces its value to within $0.005$, which is an order of magnitude smaller than any effect we report. Experiments are implemented in PyTorch with CUDA on an NVIDIA RTX 4060 GPU.

The experiments use forty-five propositions divided into three groups. Fifteen are factually true, fifteen are widely believed falsehoods like the myth that goldfish have a three-second memory, and fifteen are genuinely uncertain topics with no settled answer. Including known falsehoods is critical, as it allows us to determine if the mean drift $b$ of assumption (A2) represents actual evidence gathering or just blind assent. A drift of equal magnitude on both true and false claims cannot possibly be evidence gathering.

\subsection{Three measurement instruments}

\paragraph{E1, operator measurement.}
The first instrument directly measures the update operator $f$. To do this in practice, we feed the AI a chat history where it appears to have already stated a specific confidence percentage, which locks in its prior belief $x_i$. The simulated user then introduces the peer's opinion $x_j$ using a specific wording template, and we measure exactly how the AI updates its final answer. We built a measurement grid using the full Cartesian product of seven induced prior positions and seven peer positions. These positions are

\[
  x_i,\, x_j \;\in\; \{0.05,\ 0.20,\ 0.35,\ 0.50,\ 0.65,\ 0.80,\ 0.95\},
\]

and we wrote them into the prompt as zero-padded two-digit integer percentages. We crossed these numerical values with four message framings (Table~\ref{tab:framings} in Appendix~\ref{app:framings}) across all forty-five propositions, yielding $8{,}820$ peer measurements per model. We also included a strict control condition with no peer message, giving $315$ control measurements per model. This control is necessary because models change their stated credence even without social input, and we had to isolate the model's natural internal drift from genuine peer influence. In total, this setup produces $9{,}135$ measurements per model and $73{,}080$ across all eight models.

When calculating the realized weight $\lambda$ from Equation \eqref{eq:lambda}, the value becomes undefined when $x_j = x_i$ and gets numerically unstable when the two values are very close because the denominator vanishes. To avoid this, we computed all statistics involving $\lambda$ strictly on the subset of interactions where the agents disagreed by at least $0.30$, which leaves $5{,}040$ measurements per model. For hull escape, mean drift, and all other statistics, we used the full grid. Testing the conservation assumption (A2) requires looking at both directions of a pairwise exchange. Since $x_i$ and $x_j$ range over the exact same set of values, the grid already contains both $f(x_i, x_j)$ and $f(x_j, x_i)$, meaning no extra measurements are necessary.

\paragraph{Fitted operator.}
For one model, one claim and one framing, E1 gives the agent's new credence $x_i'$ at each of the $7\times7$ grid pairs $(x_i, x_j)$. We turn this table into a function $\hat f(x_i, x_j)$ by bilinear interpolation between grid points and clip its output to $[0,1]$. We call $\hat f$ the fitted operator. We need it because running the real models at population scale is far too costly: E3 has $1{,}620$ runs of $10^4$ agents over $30$ rounds, about $5\times10^8$ updates per model and over $60{,}000$ times the size of E2. We use $\hat f$ in two places. In E2 we run it next to the real agents to check how well it predicts them (Section~\ref{sec:fidelity}), and in E3 it stands in for the real model in every $10^4$-agent run.

\paragraph{E2, real populations.}
The second instrument places $12$ real agents on a graph and runs them for $8$ rounds. We tested complete, ring, star, and small-world topologies. By using five random seeds and four specific propositions (\texttt{t1}, \texttt{f1}, \texttt{f3}, and \texttt{u2} from Table~\ref{tab:props}), we generated $80$ distinct populations per model, which resulted in $7{,}680$ agent updates per model. We chose this specific subset of claims to represent true, false, and uncertain categories while keeping the computational cost of simulating entire populations manageable. We ran this experiment for all eight models, giving $640$ populations and $61{,}440$ agent updates in total, always with the \texttt{neutral} message framing.

To ensure a perfect comparison, we generated the exact interaction schedules and initial conditions in advance. Every agent updates synchronously by reading the state from the previous round. We then replayed these identical scenarios through three different dynamics: the real LLM agents (where inputs were rounded to integer percentages), the fitted operator $\hat f$, and classical DeGroot averaging (using a uniform weight of $\alpha = 0.5$ to represent a single observed neighbor). 

Because all three dynamics followed the exact same sequence of encounters from the exact same starting points, we can compare their trajectories agent-by-agent and round-by-round rather than just looking at overall distributions. We evaluated the results one step ahead using teacher forcing, giving each predictor the exact states the real agents held in the previous round to completely isolate single-step accuracy from cumulative errors.

\paragraph{E3, extrapolation to population scale.}
The third instrument scales our measurements up to $10{,}000$ agents. In every round, each agent draws one random neighbor to interact with, applying the fitted operator for a total of $30$ rounds. We generated $27$ different graphs using three network families (Watts-Strogatz, Erdos-Renyi, and Barabasi-Albert) across various densities, giving us algebraic connectivities ranging from $\lambda_2=0$ to $\lambda_2=5.49$. 

Against these graphs, we systematically varied the population's starting confidence structure. We swept the initial population mean across $\mu \in \{0.3, 0.5, 0.7\}$ and the stubborn agent fraction across $\text{stub} \in \{0, 0.01, 0.02, 0.05, 0.10\}$ for four specific propositions, yielding $1{,}620$ unique runs for each of the eight model operators. We anchored all committed agents strictly to the \emph{true} pole at a $0.95$ credence. Restoring them to this exact value after every round represents the absolute best-case scenario for these agents to successfully capture the rest of the population. 

We classified a run as reaching consensus if the final standard deviation dropped below $0.12$, and we labeled it as polarized if over $20$ percent of agents landed below $0.33$ while another $20$ percent landed above $0.67$. Finally, to prove exactly what drove the final outcomes, we used an ordinary least squares commonality analysis. This variance decomposition explicitly separated the predictive power of the network's connectivity (using $\lambda_2$ and mean degree) from the population's confidence structure (using the initial mean and committed fraction).

\subsection{Validating Proposition \ref{PROP:ONE}: updates leave the hull}

Table \ref{tab:one} and Figure \ref{fig:one}(a) show that agents frequently adopt positions outside the convex hull, directly violating classical convexity. To test if this actually causes the overall population diameter $D(t)$ to grow over time, we observed the real agent populations in E2. Depending on the model, the diameter increased in $5.9\%$ to $25.6\%$ of all rounds. In contrast, running the exact same interaction schedules through classical DeGroot averaging produced zero increases. This proves language model agents genuinely spread apart during interaction. Consequently, network connectivity alone is not enough to pull them back into agreement.

\begin{table*}[t]
\centering\small
\begin{tabular}{llrrrrrr}
\toprule
Model & Family & A1 esc. & strict & $\lambda{<}0$ & A2 $|$drift$|$ & $\Delta$true & $\Delta$false \\
\midrule
Qwen2.5-0.5B & Qwen & 68.3\% & 38.0\% & 21.9\% & 0.252 & +0.229 & +0.248 \\
Qwen2.5-1.5B & Qwen & 46.9\% & 10.2\% & 18.6\% & 0.144 & +0.100 & +0.096 \\
Qwen2.5-3B & Qwen & 22.3\% & 8.5\% & 8.1\% & 0.069 & +0.003 & -0.070 \\
Qwen2.5-7B$^{\dagger}$ & Qwen & 38.1\% & 12.7\% & 21.5\% & 0.086 & +0.069 & -0.025 \\
gemma-3-4b & Gemma & 34.0\% & 8.0\% & 13.2\% & 0.097 & +0.102 & -0.015 \\
gemma-2-9b$^{\dagger}$ & Gemma & 38.0\% & 4.9\% & 24.2\% & 0.076 & +0.107 & -0.040 \\
Llama-3.1-8B$^{\dagger}$ & Llama & 41.4\% & 3.7\% & 14.3\% & 0.127 & +0.187 & -0.088 \\
Mistral-7B$^{\dagger}$ & Mistral & 35.2\% & 8.5\% & 4.7\% & 0.100 & +0.050 & -0.084 \\
\midrule
\textit{classical} & n/a & 0\% & 0\% & 0\% & 0.000 & 0.000 & 0.000 \\
\bottomrule
\end{tabular}
\caption{Failure of the three assumptions across eight models; classical theory
predicts zero throughout. A1 esc.: updates leaving the hull; strict: only those
against the model's no-peer drift. $\lambda<0$: updates moving away from the peer.
A2 $|$drift$|$: mean absolute change in the pair mean; $\Delta$true, $\Delta$false:
signed drift by claim type. $^{\dagger}$Run in 4-bit.}
\label{tab:one}
\end{table*}

\begin{figure*}[t]
\centering
\includegraphics[width=\textwidth]{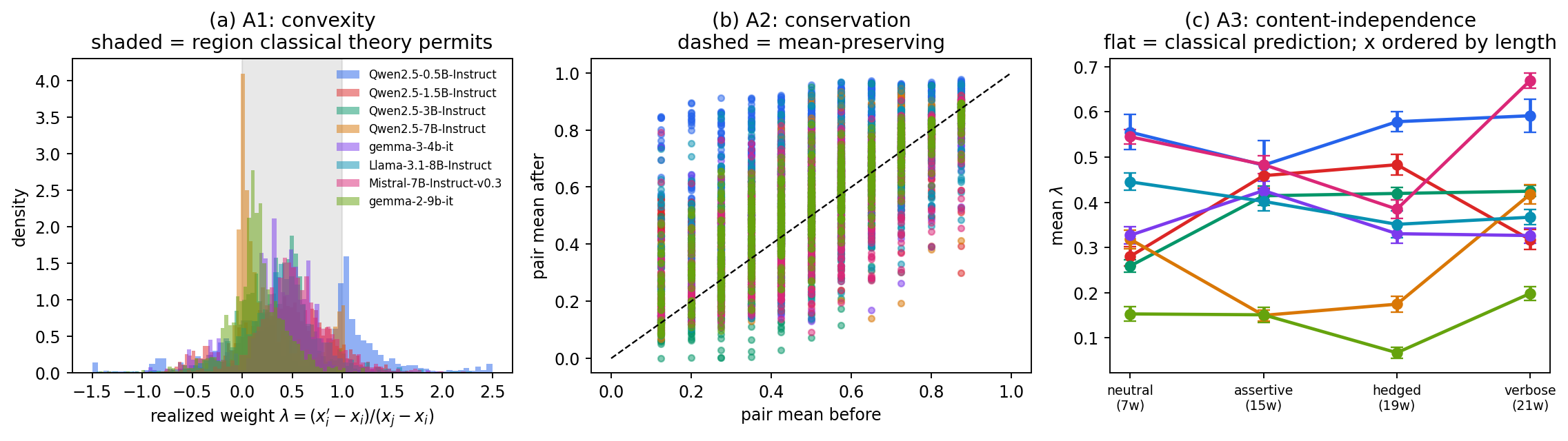}
\caption{Failure of all three structural assumptions across eight models. (a) The
distribution of the realized weight $\lambda$, with the shaded band indicating
the only values classical theory permits. (b) The pair mean before and after a
symmetric exchange, where the dashed diagonal represents what conservation
requires. (c) The mean realized weight by framing, with framings ordered by
message length, so that a pure verbosity effect would appear as a monotone
trend.}
\label{fig:one}
\end{figure*}

\subsection{Validating Corollary \ref{COR:NEG}: some updates move away from the peer}

As shown in Table \ref{tab:one}, a significant fraction of updates carry a negative realized weight across all models. While positive overshoots ($\lambda>1$) simply mean a linear model requires different parameters, a negative weight represents active repulsion. Because no non-negative row stochastic matrix can contain a negative entry under any parameterization, these specific updates prove that the classical linear model has the completely wrong functional form for describing AI agents.

\subsection{Validating Proposition \ref{PROP:TWO}: the end point stops depending on the start}

Figure \ref{fig:one}(b) highlights a systematic displacement from the mean-preserving diagonal, revealing a persistent per-exchange bias $b$ rather than true conservation. To test if this decouples the final consensus from the initial condition, we regressed the final population mean upon the starting mean $\mu$. For the heavily biased 0.5B and 1.5B operators, the slope is effectively zero. A population starting at a $0.3$ average credence terminates in the exact same place as one starting at $0.7$. The 3B model maintains classical behavior strictly because its individual bias is too small to reach saturation within thirty rounds. Qwen-7B, Gemma-4B, Mistral and Gemma-2 fall in between, with slopes from $0.17$ to $0.49$. Llama's slope is near zero ($0.02$), because it drifts strongly up on true claims and down on false ones. The real populations of E2 show the same pattern without any simulation.

\subsection{Validating Corollary \ref{COR:FALSE}: drift on false claims decides the outcome}

Figure \ref{fig:five} shows that populations of 0.5B and 1.5B agents reliably converge to confident belief in known falsehoods across every tested topology. Because their upward drift is virtually identical on both true and false claims, these networks are not gathering evidence. They are simply assenting. This blind agreement systematically drives the entire population into a confident false consensus. Qwen-7B and Gemma-4B drift less on false claims than on true ones, but their overall upward drift is still strong enough that $83\%$ and $88\%$ of their simulated runs end up believing the false claims. Their real 12-agent populations in E2 agree, finishing at credences of $0.73$ and $0.79$ on false claims. Four models reject these false claims. Llama, Mistral and Gemma-2 do so while staying high on true claims, and their real 12-agent populations finish at $0.14$, $0.47$ and $0.35$ on false claims. The 3B model also rejects them, but only because it drifts down in general, so its real populations also end below $0.5$ on the true claim. This does not go against Corollary~\ref{COR:FALSE}, which only says that a group ends up believing a false claim if $b>0$ on that claim. The four models that reject false claims all have a clearly negative drift on false claims in Table~\ref{tab:one}, from $-0.04$ to $-0.09$, so the condition of the corollary does not hold for them and their groups move down instead. Qwen-7B and Gemma-4B are close to zero, at $-0.025$ and $-0.015$, and their groups still drift up.

\begin{figure}[t]
\centering
\includegraphics[width=0.7\columnwidth]{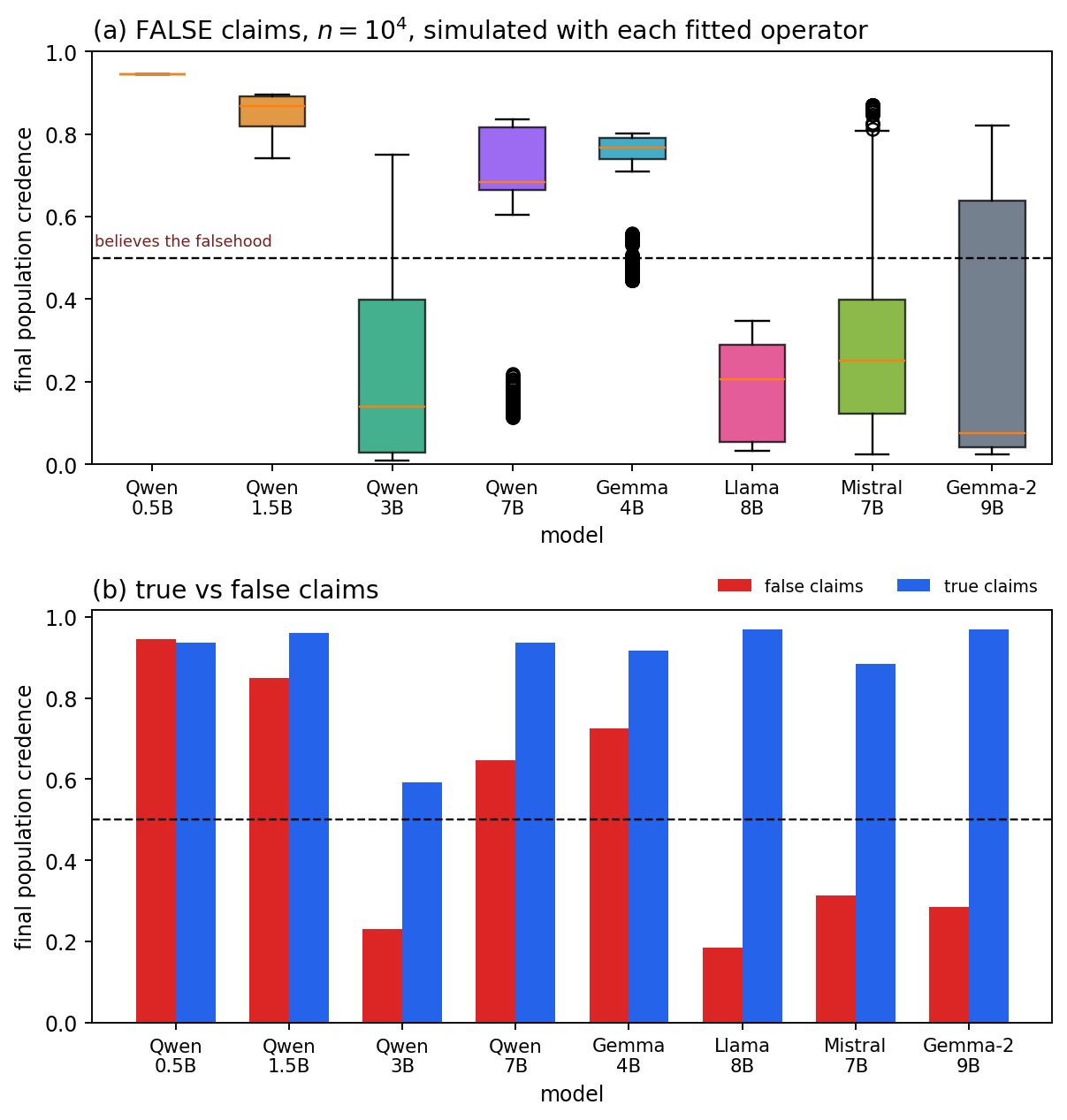}
\caption{Population outcomes at $n=10^4$. (a) Final credence on false claims for
each model. (b) True and false claims compared. Qwen2.5-3B, Llama, Mistral and
Gemma-2 end below $0.5$ on false claims.}
\label{fig:five}
\end{figure}

\subsection{Validating Proposition \ref{PROP:THREE}: framing changes the end point}

Message framing (\texttt{neutral}, \texttt{assertive}, \texttt{hedged} or \texttt{verbose}; Table~\ref{tab:framings} in Appendix~\ref{app:framings}) significantly alters the effective mathematical weight for most models, and this shift is driven by epistemic content rather than mere word count. To test the population scale impact, we ran identical networks with the exact same starting conditions and changed only the message framing. Table \ref{tab:framing} shows this single change drastically displaces the final population mean. Crucially, even tiny framing effects that are statistically undetectable in single pairwise exchanges compound over thirty rounds into massive shifts in the group's final destination. This proves that systems appearing content independent in isolated exchanges can still diverge completely as populations based solely on wording.

\begin{table}[t]
\centering\small
\setlength{\tabcolsep}{4.5pt} 
\begin{tabular}{lrrrrrr}
\toprule
& \multicolumn{4}{c}{Final mean by framing} & \multicolumn{2}{c}{Spread} \\
\cmidrule(lr){2-5}\cmidrule(lr){6-7}
Operator & neut. & asrt. & hedg. & verb. & mean & max \\
\midrule
Qwen2.5-0.5B & 0.946 & 0.777 & 0.958 & 0.953 & 0.181 & 0.304 \\
Qwen2.5-1.5B & 0.904 & 0.946 & 0.873 & 0.841 & 0.116 & 0.388 \\
Qwen2.5-3B   & 0.445 & 0.299 & 0.607 & 0.495 & 0.431 & 0.851 \\
Qwen2.5-7B   & 0.763 & 0.481 & 0.488 & 0.732 & 0.358 & 0.776 \\
gemma-3-4b   & 0.811 & 0.484 & 0.532 & 0.709 & 0.346 & 0.659 \\
Llama-3.1-8B & 0.543 & 0.350 & 0.494 & 0.542 & 0.196 & 0.564 \\
Mistral-7B   & 0.603 & 0.417 & 0.462 & 0.522 & 0.308 & 0.839 \\
gemma-2-9b   & 0.626 & 0.480 & 0.389 & 0.655 & 0.278 & 0.789 \\
\bottomrule
\end{tabular}
\caption{Final population mean under each message framing, with the graph, the
density, the seed, the initial condition, and the neighbour draws held identical
across framings. The per cell spread is the range across the four framings
within a single sweep cell, and is therefore a matched comparison rather than a
difference of averages.}
\label{tab:framing}
\end{table}

\subsection{Validating the fitted operator}
\label{sec:fidelity}

To justify using the fitted operator for massive scale simulations, we validated its accuracy against real agent populations. Evaluated one step ahead, the fitted operator consistently outperforms classical DeGroot averaging across all topologies and all eight models, with DeGroot's error between $1.28$ and $4.05$ times larger (Figure~\ref{fig:two}(b) and (c) show this for the 3B model). While DeGroot approximates single steps reasonably well, its true failure lies in predicting the final destination. As Figure \ref{fig:two}(a) shows for the 3B model, when real agents discuss a false claim, they correctly lower their credence to reject it. DeGroot, however, mathematically forces the trajectory upward to preserve the initial mean. A classical model can therefore predict individual moves accurately but still drive the group to an entirely wrong consensus by enforcing false mathematical invariants.

\begin{figure*}[t]
\centering
\includegraphics[width=\textwidth]{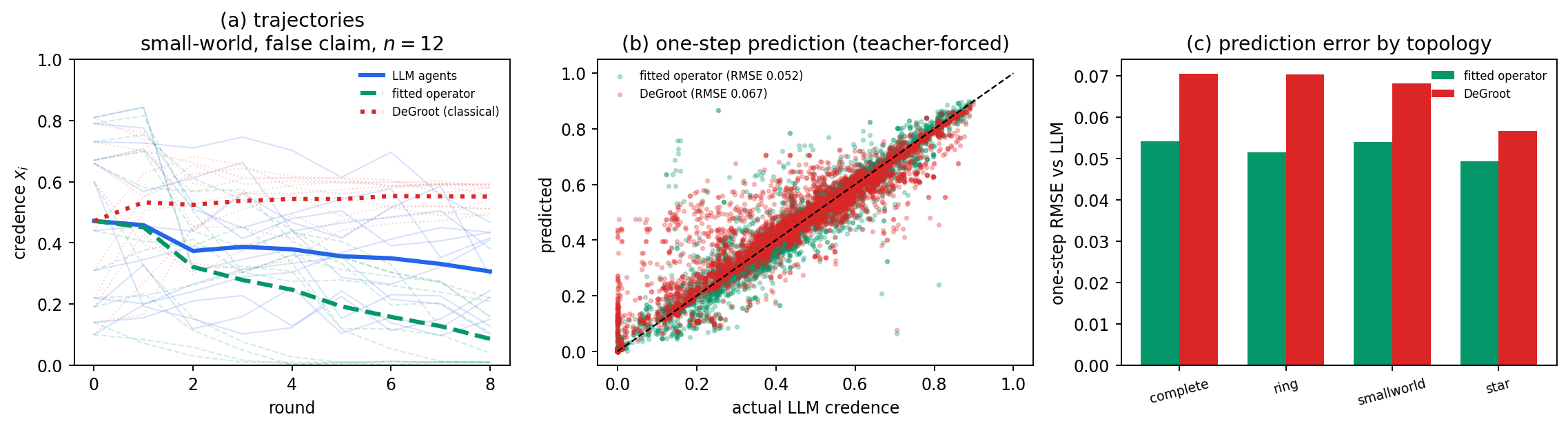}
\caption{Validation against eighty populations of twelve real Qwen2.5-3B agents. (a)
Trajectories on a false claim, in which DeGroot rises toward the conserved
initial mean while the real population falls. (b) One step teacher forced
prediction against ground truth. (c) Prediction error by topology.}
\label{fig:two}
\end{figure*}

\subsection{Validating Section \ref{SEC:REPLACE}: connectivity explains almost none of the outcome}

To test whether predictive power still resides in the graph, we decomposed the outcome variance into factors driven by connectivity versus confidence structure. For all eight operators, the network's algebraic connectivity ($\lambda_2$) and mean degree uniquely explain zero percent of the final population mean. The initial confidence structure explains more, but how much depends on the model, ranging from almost nothing to $0.23$ of the final mean. The real populations of E2 point the same way. Changing the topology moves a model's final credence on false claims by at most $0.11$, while changing the model moves it by up to $0.82$. The central metric of classical consensus theory therefore carries essentially no predictive information for AI agent populations.

\subsection{Committed agents fragment rather than capture}

In models of human opinion dynamics, a committed minority typically steers the population toward its own position. We observe no such effect here. Raising the fraction of stubborn agents to ten percent moves the final population mean by at most $0.15$. Instead, these committed agents simply prevent agreement. For the 3B model, the consensus rate drops from $90.7\%$ to $16.7\%$ with no increase in polarization. Qwen-7B and Gemma-4B show the same effect more mildly, with consensus dropping from $91.4\%$ to $70.4\%$ and from $100\%$ to $83.3\%$, while their final means move by only about $0.03$. Llama, Mistral and Gemma-2 show the same pattern, with consensus dropping from $100\%$ to $0\%$, from $91.4\%$ to $9.9\%$, and from $90.7\%$ to $25.9\%$. Stubborn AI agents disperse and fragment the population rather than successfully capturing it.

\section{Conclusion}

Classical consensus theory fails to describe populations of LLM agents because all three structural assumptions break down under direct measurement. The standard non-linear repairs for human populations also fail, because they still keep every update inside the convex hull. The practical consequence is that predicting a multi-agent system's behavior by analyzing its network topology is fundamentally flawed. Algebraic connectivity explains almost none of the variation in where a population lands. Instead, system dynamics are driven by the chosen model's internal update bias and the specific wording of the prompts. A designer choosing between different network graphs is adjusting a variable that hardly matters.

This failure highlights a concrete deployment risk. Populations of four of our eight models, including Qwen2.5 at 7B, reliably talk themselves into confident agreement on objectively false claims, regardless of how their network is wired. Small upward drifts in isolated exchanges compound until the entire network saturates in collective delusion. This does not go away with model size. The 7B Qwen model's populations believe false claims, while those of the smaller 3B Qwen model and of Llama-8B, Mistral-7B and Gemma-2-9B do not, so the outcome depends on each model's own drift rather than on how large it is.

\section{Limitations}

Two main limitations bound these findings. First, our tested models span 0.5B to 9B parameters, meaning these findings may not perfectly extrapolate to frontier models like GPT, Claude or Gemini, which are far larger and whose token probabilities are often not fully available, so our readout cannot be applied to them directly. Second, our state variable relies on a simple scalar credence for a single proposition, which lacks the multidimensional complexity of beliefs in deployed systems.

\section{Broader Impact}
\label{SEC:BROADER-IMPACT}

This work studies how groups of LLM agents update their beliefs when they talk to each other, a setup now common in multi-agent debate, voting, and agent teams. We show that how the agents are connected says little about where the group ends up, and that some models reach confident agreement on false claims. This helps designers of such systems avoid treating group agreement as a sign of correctness, and gives them a simple pairwise test to check a model's drift on known false claims before using it in a group. All experiments use public open-weight models and public factual claims, and involve no human participants or personal data.

\subsection*{Acknowledgment: GenAI Use}
ChatGPT~\cite{chatgpt} has been used to improve the syntax and grammar of several sections in the manuscript. It has also been used to identify and review literature on multi-agent systems, opinion dynamics, and distributed consensus theory.

\bibliography{neurips_2025}
\bibliographystyle{plainnat}

\newpage
\appendix

\section{Proofs of the claims}

\subsection{Proof of Proposition~\ref{PROP:ONE}: Hull escape destroys monotone contraction}
\label{PROOF:PROP:ONE}
\begin{proof}
This proposition demonstrates how agents adopting positions outside their visible convex hull destroy the monotone contraction of the group's diameter.

Let the diameter of the population at time $t$ be defined as $D(t) = \max_k x_k(t) - \min_k x_k(t)$. 
During an update step, agent $i$ observes a subset of states $\mathcal{N}_i \cup \{i\}$. The convex hull of these visible states is bounded by the local minimum and local maximum:
$$[\min_{j \in \mathcal{N}_i \cup \{i\}} x_j(t), \max_{j \in \mathcal{N}_i \cup \{i\}} x_j(t)]$$
Because these local extremes are drawn from the global population, they are strictly bounded by the global extremes:
$$\min_k x_k(t) \le \min_{j \in \mathcal{N}_i \cup \{i\}} x_j(t) \quad \text{and} \quad \max_{j \in \mathcal{N}_i \cup \{i\}} x_j(t) \le \max_k x_k(t)$$
By the premise of the proposition, an agent's new state $x_i(t+1)$ can escape its local convex hull by at most a margin of $\delta > 0$. This establishes the absolute upper and lower bounds for the updated state:
$$x_i(t+1) \le \max_{j \in \mathcal{N}_i \cup \{i\}} x_j(t) + \delta \le \max_k x_k(t) + \delta$$
$$x_i(t+1) \ge \min_{j \in \mathcal{N}_i \cup \{i\}} x_j(t) - \delta \ge \min_k x_k(t) - \delta$$
The new diameter at time $t+1$ is $D(t+1) = \max_i x_i(t+1) - \min_i x_i(t+1)$. Substituting the worst case bounds yields:
$$D(t+1) \le (\max_k x_k(t) + \delta) - (\min_k x_k(t) - \delta)$$
$$D(t+1) \le \max_k x_k(t) - \min_k x_k(t) + 2\delta$$
$$D(t+1) \le D(t) + 2\delta$$
This bound is perfectly attained if at least one agent updates to the maximum possible positive deviation while another independently updates to the maximum possible negative deviation in the same round. Consequently, $D(t)$ is no longer monotonically decreasing, meaning the network is not guaranteed to contract toward agreement.
\end{proof}

\subsection{Proof of Corollary~\ref{COR:NEG}: Negative weights are unrepresentable}
\label{PROOF:COR:NEG}
\begin{proof}
This corollary proves that a realized negative weight fundamentally breaks the row stochastic matrix requirement of the classical linear model.

The classical consensus model dictates that the updated state $x_i'$ is a weighted sum of the states the agent observes:
$$x_i' = \sum_{k} W_{ik} x_k(t)$$
For a controlled pairwise exchange between agent $i$ and a peer $j$, the only visible states are $x_i$ and $x_j$. The update equation simplifies to:
$$x_i' = W_{ii} x_i + W_{ij} x_j$$
The classical model strictly requires the weight matrix $W$ to be row stochastic, meaning all entries are non-negative ($W_{ik} \ge 0$) and each row sums to exactly $1$ ($\sum_k W_{ik} = 1$). For this isolated pairwise interaction, $W_{ii} + W_{ij} = 1$, allowing us to substitute $W_{ii} = 1 - W_{ij}$:
$$x_i' = (1 - W_{ij}) x_i + W_{ij} x_j$$
$$x_i' = x_i - W_{ij} x_i + W_{ij} x_j$$
$$x_i' = x_i + W_{ij} (x_j - x_i)$$
Isolating the weight applied to the peer's state gives:
$$W_{ij} = \frac{x_i' - x_i}{x_j - x_i}$$
This isolated term matches the exact definition of the realized weight $\lambda(x_i, x_j, c_j)$. Therefore, $W_{ij} = \lambda$.
If the empirically measured update yields $\lambda < 0$, it dictates that $W_{ij} < 0$. A negative entry is explicitly forbidden by the definition of a non-negative row stochastic matrix, proving that no such matrix can reproduce this specific update for any choice of weights.
\end{proof}

\subsection{Proof of Proposition~\ref{PROP:TWO}: Operator bias relocates the fixed point}
\label{PROOF:PROP:TWO}
\begin{proof}
This proposition proves mathematically that any persistent bias in the update rule dictates the system's final destination, rendering the initial conditions irrelevant over a sufficiently long time horizon.

Let $m(t)$ represent the population mean at time step $t$. By assumption, each exchange shifts the mean by an expected bias $b > 0$, meaning the recurrence relation for the unbounded mean is:
$$m(t+1) = m(t) + b$$
We can solve this linear recurrence relation directly by induction. For $t=1$:
$$m(1) = m(0) + b$$
For $t=2$:
$$m(2) = m(1) + b = (m(0) + b) + b = m(0) + 2b$$
Assuming the relation holds for an arbitrary step $t$, we have:
$$m(t) = m(0) + tb$$
However, the state variable $x_i(t)$ represents a credence probability and is strictly bounded within the interval $[0, 1]$. Therefore, the population mean $m(t)$ must also remain interior to $[0, 1]$. Applying this absolute upper bound gives the true state at time $T$:
$$m(T) = \min\{1, m(0) + Tb\}$$
To prove that the population eventually sits arbitrarily close to $1$ independently of $m(0)$, we solve for the horizon $T$ required to reach the maximum bound:
$$1 \le m(0) + Tb$$
$$1 - m(0) \le Tb$$
$$T \ge \frac{1 - m(0)}{b}$$
Because $b > 0$ and $m(0) \in [0, 1]$, the term $\frac{1 - m(0)}{b}$ is always a finite, non-negative number. This proves there is always a finite time horizon $T$ where the mean will saturate at $1$, completely overriding the initial mean $m(0)$.
\end{proof}

\subsection{Proof of Corollary~\ref{COR:FALSE}: Confident false consensus}
\label{PROOF:COR:FALSE}
\begin{proof}
This corollary applies the mathematical limit established in Proposition~\ref{PROP:TWO} to the specific context of factuality and collective delusion.

Let the proposition under discussion be objectively false. In a purely evidence-gathering system, the population mean $m(t)$ should ideally converge toward $0$. 

However, suppose the empirical update operator exhibits a strictly positive bias ($b > 0$) for this false proposition. By Proposition~\ref{PROP:TWO}, the system's dynamics are governed by:
$$m(T) = \min\{1, m(0) + Tb\}$$
Even if the system designers initialize the agents with accurate, highly skeptical prior beliefs (e.g., setting $m(0)$ close to $0$), the term $Tb$ strictly increases as interactions occur. The initial condition $m(0)$ merely dictates how many interaction rounds $T$ it will take for the system to saturate. 

Because the bias $b$ adds positive drift regardless of the proposition's truth value, the population mean is structurally forced to approach $1$. Thus, the population will invariably converge to a high credence state, resulting in a confident but entirely false consensus that no choice of starting conditions can prevent.
\end{proof}

\subsection{Derivation of Proposition~\ref{PROP:THREE}: Content dependence underdetermines the graph}
\label{PROOF:PROP:THREE}
\begin{proof}
This proposition demonstrates mathematically that if communication weights rely on message framing, the network topology alone is strictly insufficient to predict the dynamic trajectory or final consensus of the system.

In classical consensus theory, the weight matrix $W$ is strictly a function of the graph topology $\mathcal{G} = (\mathcal{V}, \mathcal{E})$. We can denote this invariant matrix as $W_{\mathcal{G}}$. The state update for the entire network is written as:
$$x(t+1) = W_{\mathcal{G}} \, x(t)$$
Given an initial condition $x(0)$, the trajectory of the system is entirely deterministic and defined by:
$$x(t) = W_{\mathcal{G}}^t \, x(0)$$
However, empirical measurements dictate that the realized weight $\lambda$ is a function of the message framing $c_j$. As established in Corollary~\ref{COR:NEG}, the realized weight corresponds directly to the matrix entry: $W_{ij} = \lambda(x_i, x_j, c_j)$. 

Because $W_{ij}$ depends on the continuous framing variables vector $c(t)$ at time $t$, the weight matrix is no longer a static property of the graph $\mathcal{G}$. It becomes a time-varying, content-dependent matrix $W(c(t))$. The state update equation becomes:
$$x(t+1) = W(c(t)) \, x(t)$$
To prove underdetermination, consider two separate multi-agent deployments, System A and System B. Both share the exact same network topology $\mathcal{G}$ and the exact same initial state $x(0)$. 

Let System A adopt a communication convention using framing $c_A$ (e.g., highly assertive framing), and System B adopt framing $c_B$ (e.g., highly hedged framing), where $c_A \neq c_B$. At time step $t=1$, their respective updates are:
$$x_A(1) = W(c_A) \, x(0)$$
$$x_B(1) = W(c_B) \, x(0)$$
Because the measured realized weight varies significantly with framing, $W(c_A) \neq W(c_B)$. Assuming a non-trivial $x(0)$, this guarantees that:
$$x_A(1) \neq x_B(1)$$
Despite possessing identical wiring diagrams and starting beliefs, the trajectories of the two systems diverge immediately at the first timestep. Consequently, they will cascade toward completely different fixed points. Therefore, no mathematical function $f(\mathcal{G})$ that maps only the graph topology to a consensus outcome can correctly predict the fixed points of both System A and System B. The graph fundamentally underdetermines the dynamics.
\end{proof}

\section{Experimental Details}
\label{app:repro}

This appendix gives every element needed to reproduce the three experiments,
namely the exact propositions, the exact message framings, the exact prompts,
the readout procedure, etc.

\subsection{Propositions}
\label{app:props}

We use forty-five propositions split into three groups: fifteen facts, fifteen common myths, and fifteen genuinely uncertain statements. We intentionally included known falsehoods to figure out if the models' internal drift actually tracks truth or just represents blind agreement. If a model drifts the same amount on both true and false claims, it is definitely not gathering evidence.

\begin{table*}[t]
\centering\small
\begin{tabular}{llp{0.80\linewidth}}
\toprule
ID & Truth & Proposition text (verbatim) \\
\midrule
\texttt{t1} & true & ``Water expands when it freezes.'' \\
\texttt{t2} & true & ``Iceland has no native mosquito population.'' \\
\texttt{t3} & true & ``The human genome contains fewer than 30,000 protein-coding genes.'' \\
\texttt{t4} & true & ``Venus rotates on its axis in the opposite direction to most planets.'' \\
\texttt{t5} & true & ``Octopuses have three hearts.'' \\
\texttt{t6} & true & ``Bananas are botanically classified as berries.'' \\
\texttt{t7} & true & ``Sound travels faster in water than in air.'' \\
\texttt{t8} & true & ``Sharks existed on Earth before trees did.'' \\
\texttt{t9} & true & ``Wombats produce cube-shaped droppings.'' \\
\texttt{t10} & true & ``The University of Oxford is older than the Aztec Empire.'' \\
\texttt{t11} & true & ``The Eiffel Tower is taller in summer than in winter.'' \\
\texttt{t12} & true & ``Humans share more than 95 percent of their DNA with chimpanzees.'' \\
\texttt{t13} & true & ``Cleopatra lived closer in time to the Moon landing than to the building of the Great Pyramid of Giza.'' \\
\texttt{t14} & true & ``The Pacific Ocean is larger than all of Earth's land area combined.'' \\
\texttt{t15} & true & ``Light from the Sun takes about eight minutes to reach Earth.'' \\
\midrule
\texttt{f1} & false & ``The Great Wall of China is visible to the naked eye from the Moon.'' \\
\texttt{f2} & false & ``Goldfish have a memory span of only three seconds.'' \\
\texttt{f3} & false & ``Humans use only 10 percent of their brains.'' \\
\texttt{f4} & false & ``Lightning never strikes the same place twice.'' \\
\texttt{f5} & false & ``Bulls are enraged by the color red.'' \\
\texttt{f6} & false & ``Bats are blind.'' \\
\texttt{f7} & false & ``Shaving makes hair grow back thicker.'' \\
\texttt{f8} & false & ``Napoleon Bonaparte was unusually short for a man of his time.'' \\
\texttt{f9} & false & ``Different parts of the tongue detect different basic tastes.'' \\
\texttt{f10} & false & ``Ostriches bury their heads in the sand when frightened.'' \\
\texttt{f11} & false & ``Vikings wore horned helmets in battle.'' \\
\texttt{f12} & false & ``Eating sugar makes children hyperactive.'' \\
\texttt{f13} & false & ``A penny dropped from the top of the Empire State Building can kill a person.'' \\
\texttt{f14} & false & ``Dogs see only in black and white.'' \\
\texttt{f15} & false & ``Cracking your knuckles causes arthritis.'' \\
\midrule
\texttt{u1} & uncertain & ``More than half the world's population will live in cities of over one million people by 2050.'' \\
\texttt{u2} & uncertain & ``A randomly selected adult in the United States owns a valid passport.'' \\
\texttt{u3} & uncertain & ``Nuclear power will supply more than 15 percent of global electricity in 2040.'' \\
\texttt{u4} & uncertain & ``The majority of undiscovered mammal species are rodents or bats.'' \\
\texttt{u5} & uncertain & ``Humans will land on Mars before 2040.'' \\
\texttt{u6} & uncertain & ``The global average temperature in 2050 will be more than 2 degrees Celsius above pre-industrial levels.'' \\
\texttt{u7} & uncertain & ``Life exists or once existed on Mars.'' \\
\texttt{u8} & uncertain & ``More than half of all new cars sold worldwide in 2035 will be fully electric.'' \\
\texttt{u9} & uncertain & ``The world population will peak before 2080.'' \\
\texttt{u10} & uncertain & ``A treatment that cures Alzheimer's disease will be approved before 2040.'' \\
\texttt{u11} & uncertain & ``The Voynich manuscript is written in a real, meaningful language.'' \\
\texttt{u12} & uncertain & ``There are more than one trillion species of microbes on Earth.'' \\
\texttt{u13} & uncertain & ``Intelligent life exists elsewhere in the Milky Way.'' \\
\texttt{u14} & uncertain & ``Dark matter particles will be directly detected before 2045.'' \\
\texttt{u15} & uncertain & ``Average human life expectancy worldwide will exceed 80 years by 2060.'' \\
\bottomrule
\end{tabular}
\caption{The forty-five propositions, verbatim. E1 uses all forty-five. E2 and E3 use
the four marked \texttt{t1}, \texttt{f1}, \texttt{f3}, \texttt{u2}.}
\label{tab:props}
\end{table*}

\subsection{Message framings}
\label{app:framings}

A framing dictates how a peer's credence is phrased while keeping the actual numeric value ($x_j$) exactly the same. The prompt replaces the \texttt{\{xj\}} token with a two-digit integer percentage representing that peer's stated confidence. We included the \texttt{verbose} framing specifically to separate a message's stated confidence from its overall word count. Since a confident message naturally tends to be longer, the verbose framing provides a long but completely neutral message. If word count alone drove the content effect, the verbose prompt would behave like the marked framings instead of the neutral one.

\begin{table}[h]
\centering\small
\begin{tabular}{llp{0.56\linewidth}}
\toprule
Framing & Words & Template (verbatim) \\
\midrule
\texttt{neutral} & 7 & Another participant said: ``My estimate is \{xj\}\%.'' \\[2pt]
\texttt{assertive} & 15 & Another participant, who seemed highly confident, said: ``I'm certain about this one. It is \{xj\}\%.'' \\[2pt]
\texttt{hedged} & 19 & Another participant, who seemed quite unsure, said: ``I really don't know, I'm just guessing here, but maybe around \{xj\}\%.'' \\[2pt]
\texttt{verbose} & 21 & Another participant, who was also taking part in this same estimation task today, said: ``My estimate for this one is \{xj\}\%.'' \\[2pt]
\texttt{control} & --- & No peer message is shown. See Section \ref{app:prompts}. \\
\bottomrule
\end{tabular}
\caption{Message framings. All four peer framings carry an identical numeric
value and differ only in wording.}
\label{tab:framings}
\end{table}

\subsection{Exact prompts}
\label{app:prompts}

Each measurement consists of a four-turn chat. By placing the prior position $x_i$ in the assistant's turn, we force the agent to adopt the position as its own stated belief rather than having it assigned externally. The final user turn then introduces the peer's message. In the control condition, this turn simply asks the model to reconsider its estimate without providing any social input. This control is crucial because it isolates the model's natural internal drift from genuine peer influence, which directly enables the strict hull escape measurements discussed in Section \ref{sec:exp}.

\paragraph{System turn.}
\begin{verbatim}
You are a participant in a group
estimation task. You give calibrated
probability estimates as percentages.
You answer with a number only.
\end{verbatim}

\paragraph{Turn structure.} With \texttt{\{P\}} the proposition text, \texttt{\{xi\}} the induced prior as a zero-padded two-digit integer, and \texttt{\{peer\}} the rendered framing from Table \ref{tab:framings}:

\begin{verbatim}
[system] <system turn above>

[user]   Claim: "{P}"

         What is your confidence that
         this claim is TRUE? Answer
         with a two-digit percentage
         from 00 to 99.

[assistant] {xi}

[user]   {peer}

         What is your updated confidence
         that the claim is TRUE? Answer
         with a two-digit percentage
         from 00 to 99.
\end{verbatim}

\paragraph{Control condition.} The final user turn is replaced by:

\begin{verbatim}
Please reconsider your estimate.

What is your updated confidence that
the claim is TRUE? Answer with a 
two-digit percentage from 00 to 99.
\end{verbatim}

We render these four turns into token IDs using each model's specific chat template with \\
\texttt{add\_generation\_prompt=True}. If a model's template rejects a dedicated \texttt{system} role, as Gemma does, we simply fold the system instructions into the first user turn while keeping the rest of the chat identical.

\subsection{Credence readout}
\label{app:readout}

We intentionally avoid text sampling. Sampling introduces random noise that looks identical to genuine variations in the model's update rule. Furthermore, parsing sampled text often fails on malformed outputs, which would silently bias our dataset toward contexts the model happens to answer cleanly. We avoid these issues by computing the exact joint distribution over a two-digit answer and taking its mathematical expectation, as shown in Equation \eqref{eq:readout}. The full procedure is:

\begin{enumerate}[leftmargin=1.4em,itemsep=1pt,topsep=3pt]
\item Render the four turn chat to token ids with the model's chat template and
      \texttt{add\_generation\_prompt=True}.
\item Form ten sequences by appending each digit token \texttt{"0"} through
      \texttt{"9"}. At construction we assert that each of the ten digits encodes
      to exactly one token under the model's tokenizer, which holds for the Qwen
      and Gemma models.
\item Run one batched forward pass over those ten sequences using left padding,
      with explicit position ids computed as
      \texttt{(attention\_mask.cumsum(-1) - 1).clamp(min=0)}. Left padding is
      what keeps the final two positions aligned across contexts of different
      lengths.
\item Read $p(d_1 \mid \text{ctx})$ from the logits at position $-2$. These are
      identical across the ten sequences because those sequences share the
      entire prefix, so one row suffices. Read $p(d_2 \mid \text{ctx}, d_1)$ from
      position $-1$, which gives one row per value of $d_1$.
\item Renormalise both distributions over the ten digit tokens and return
      $\frac{1}{99}\sum_{d_1}\sum_{d_2}(10 d_1 + d_2)\, p(d_1)\, p(d_2 \mid d_1)$.
\end{enumerate}

The steps above are for Qwen and Gemma. Llama's tokenizer stores every number from 00 to 99 as a single token, so for Llama we skip steps 2 to 5. We run one forward pass on the prompt alone, read the probability of each of the 100 tokens \texttt{"00"} to \texttt{"99"} at the last position, renormalise over them, and return $\frac{1}{99}\sum_{k=0}^{99} k\, p(k)$. Mistral's tokenizer puts a space token before a number, so for Mistral each of the ten sequences in step 2 is the prompt, then the space token, then $d_1$, and steps 3 to 5 stay the same. We add the space token only if, at the answer position, the model puts more probability on it than on the ten digits together, which we check once on the first prompt. For both models, the readout tracks the number the model actually writes (correlation $0.90$ for Llama and $0.87$ for Mistral).

We materialize only the final two positions by setting \texttt{logits\_to\_keep=2}. Computing logits for the entire sequence would allocate billions of floating-point values and exceed several gigabytes of memory, which would force us to use an impractically small batch size.

\end{document}